\documentclass[11pt,a4paper]{article}

\usepackage[T1]{fontenc}
\usepackage{lmodern}
\usepackage[margin=27mm]{geometry}
\usepackage{amsmath,amssymb,amsthm}
\usepackage{microtype}
\usepackage{xcolor}
\usepackage[colorlinks=true,
            linkcolor=blue!50!black,
            citecolor=blue!50!black,
            urlcolor=blue!50!black,
            breaklinks=true]{hyperref}

\newcommand{\doilink}[1]{\href{https://doi.org/#1}{doi:\nolinkurl{#1}}}

\newtheorem{theorem}{Theorem}
\newtheorem{lemma}[theorem]{Lemma}

\theoremstyle{definition}

\theoremstyle{remark}

\title{Propositional formula minimization is \(\Sigma_2^p\)-complete \\ under many-one reductions}
\author{Reijo Jaakkola}
\date{25.10.2026}

\begin{document}
\maketitle

\begin{abstract}
The Minimum Equivalent Expression problem asks, given a
propositional formula \(\phi\) and an integer \(k\), whether
\(\phi\) is equivalent to a formula with at most \(k\)
occurrences of propositional variables. While the problem
is clearly in \(\Sigma_2^p\), its completeness under
polynomial-time many-one reductions has remained open
since at least 1979. Buchfuhrer and Umans showed in 2008
that it is \(\Sigma_2^p\)-complete under polynomial-time
Turing reductions. In this note, we present a proof
discovered by GPT-Sol-6.1 which establishes that
MEE is \(\Sigma_2^p\)-complete also under polynomial-time
many-one reductions.
\end{abstract}

\section{Introduction}

The Minimum Equivalent Expression (MEE) problem is a fundamental computational problem in logic. Given a propositional formula
\(\phi\) and an integer \(k\), the problem asks whether there exists a
formula equivalent to \(\phi\) containing at most \(k\) occurrences of
propositional variables. The problem belongs to \(\Sigma_2^p\), since
one can existentially guess a formula of the prescribed size and then
verify its equivalence to \(\phi\) using a coNP oracle.

Determining the exact complexity of MEE has a long history. The problem
was already listed as open by Garey and Johnson in
1979~\cite{GareyJohnson1979}. In 1997 Hemaspaandra and Wechsung established hardness under polynomial-time
many-one reductions for the class of problems solvable using parallel
access to an NP oracle~\cite{HemaspaandraWechsung2002}.
Subsequently, Buchfuhrer and Umans proved in 2008 that MEE is
\(\Sigma_2^p\)-complete under polynomial-time Turing
reductions~\cite{BuchfuhrerUmans2008,BuchfuhrerUmans2011}.
Their result also established the corresponding completeness for
bounded-depth formulas of every fixed depth at least three. The case of
depth two, which amounts to minimizing DNF formulas, was proved earlier to be
\(\Sigma_2^p\)-complete under many-one reductions in Umans~\cite{Umans2001}.

In this note, we resolve this question by proving that MEE is also
\(\Sigma_2^p\)-complete under polynomial-time many-one reductions. At a high level, our strategy follows that of Buchfuhrer and
Umans~\cite{BuchfuhrerUmans2011}. Like them, we reduce from Succinct
Set Cover, a \(\Sigma_2^p\)-complete problem introduced by
Umans~\cite{Umans1999}. In their construction, the size threshold of
the final MEE instance depends on the minimum size of an equivalent
formula for an auxiliary function. They have no explicit expression for
this quantity, so they determine it with oracle calls to MEE, and this
is why their reduction is a Turing reduction. It turns out that by using technical tricks we can modify their reduction so that the threshold can be written down explicitly, removing the need for oracle calls.

Our motivation for studying MEE comes from explainable AI (XAI), and in
particular from our earlier work on logic-based explanations~\cite{JaakkolaEtAl2025}.
There we distinguish between \emph{local} explanations, which explain the
behavior of a model on a single input, and \emph{global} explanations, which
explain its behavior on all inputs. If a classifier is represented by a
propositional formula \(\phi\) and the explanation is required to be completely
faithful, a global explanation of \(\phi\) is a formula \(\psi\) of minimal size
that is equivalent to \(\phi\).\footnote{At first sight, it may seem odd to define an optimal explanation of a
propositional formula as merely a smallest equivalent formula: how can explaining
be reduced to optimizing size? The requirement of faithfulness does the real work
here. A faithful explanation must describe fully how the formula behaves on all
inputs, and so it can itself be interpreted as a formula computing the same function. The only thing
left to improve is how easy the explanation is to understand, and one natural way
to measure this is its size.} When the size of a formula is measured by the number of occurrences of propositional variables, the problem of finding an optimal global
explanation is precisely MEE.

\medskip

\noindent\textbf{Use of AI.}
The proof presented in this note was discovered by GPT-Sol-6.1.
Opus 5.5 was subsequently used to rewrite the proof
to improve its clarity and readability. GPT-6 and Sonnet 5.5 assisted
with the writing and editing of the manuscript.

\section{Preliminaries}

Given a set \(\Phi\) of propositional variables, we define \(\mathrm{PL}[\Phi]\) to be the set of propositional (or Boolean) formulas over \(\Phi\) generated by the following grammar.
\[\phi ::= \top \mid \bot \mid p \mid \neg \phi \mid \phi \land \phi \mid \phi \lor \phi\]
Here \(p \in \Phi\).

A \(\Phi\)-\textbf{assignment} is a function \(v:\Phi \to \{0,1\}\). Given an assignment \(v\), a variable \(p\in\mathrm{dom}(v)\), and \(b\in\{0,1\}\), we write \(v[p/b]\) for the assignment
that agrees with \(v\) on all variables other than \(p\)
and satisfies \(v[p/b](p)=b\). The \textbf{truth value} of \(\phi\) under assignment \(v\) is defined as usual and we denote it by \(v(\phi)\). Given two formulas \(\phi,\psi \in \mathrm{PL}[\Phi]\), we write \(\phi \models \psi\) if \(v(\phi) = 1\) implies \(v(\psi) = 1\) for every \(\Phi\)-assignment \(v\). Furthermore, we write \(\phi \equiv \psi\) if \(\phi \models \psi\) and \(\psi \models \phi\).

Let \(\Phi\) be a set of propositional variables and \(w:\Phi \to \mathbb{Z}_+\) a function. Think of \(w\) as assigning a weight to each proposition symbol. Given \(\phi \in \mathrm{PL}[\Phi]\), we define the \(w\)-\textbf{size} \(\mathrm{size}_w(\phi)\) of \(\phi\) recursively as follows.
\begin{itemize}
    \item \(\mathrm{size}_w(\top) = \mathrm{size}_w(\bot) = 0\).
    \item \(\mathrm{size}_w(p) = w(p)\), for \(p \in \Phi\).
    \item \(\mathrm{size}_w(\neg \phi) = \mathrm{size}_w(\phi)\).
    \item \(\mathrm{size}_w(\phi \land \psi) = \mathrm{size}_w(\phi) + \mathrm{size}_w(\psi)\).
    \item \(\mathrm{size}_w(\phi \lor \psi) = \mathrm{size}_w(\phi) + \mathrm{size}_w(\psi)\).
\end{itemize}
With this size convention, we can assume without loss of generality that our formulas are in negation normal form. That is, negation occurs only in front of proposition variables.

We also define
\[L_w(\phi) := \min\{\mathrm{size}_w(\psi) \mid \phi \equiv \psi\}.\]
When \(w\) is the constant function which maps every proposition variable to one, we drop \(w\) and write \(\mathrm{size}(\phi)\) and \(L(\phi)\) instead of \(\mathrm{size}_w(\phi)\) and \(L_w(\phi)\) respectively. We are interested in the following computational problem.

\begin{center}
\fbox{%
\begin{minipage}{0.92\linewidth}
\textsc{Minimum Equivalent Expression (MEE)}

\medskip
\textbf{Input:} A propositional formula $\phi$ and an integer $k \geq 0$.

\smallskip
\textbf{Question:} Is $L(\phi) \leq k$?
\end{minipage}%
}
\end{center}

\noindent For $w:\Phi\to\mathbb{Z}_+$ and $\phi\in\mathrm{PL}[\Phi]$,
write $\phi^w$ for the formula obtained by replacing each variable $p$
with a conjunction of $w(p)$ fresh variables. We use the same block
for every occurrence of $p$ and disjoint blocks for distinct variables. The following lemma is essentially proved in~\cite{BuchfuhrerUmans2011}.

\begin{lemma}\label{lem:leaf-weighting}
Let $w:\Phi\to\mathbb{Z}_+$ and $\phi\in\mathrm{PL}[\Phi]$. Then
\[
L(\phi^w)=L_w(\phi).
\]
\end{lemma}

\section{Main result}

In this section we prove the following.

\begin{theorem}\label{thm:main}
\(\mathrm{MEE}\) is \(\Sigma_2^p\)-complete under polynomial-time many-one reductions.
\end{theorem}

\noindent The upper bound is trivial. For the lower bound we use the following problem.

\begin{center}
\fbox{%
\begin{minipage}{0.92\linewidth}
\textsc{Succinct Set Cover (SSC)}

\medskip
\textbf{Input:} Satisfiable 3-DNF formulas
$\theta_1,\ldots,\theta_n \in \mathrm{PL}[\Phi]$
such that
\[\bigvee_{i=1}^n \theta_i \equiv\top,\]
and an integer $1 \leq k < n$.

\smallskip
\textbf{Question:} Is there a set $I\subseteq\{1,\ldots,n\}$
with $|I|\leq k$ such that
\[
\bigvee_{i\in I} \theta_i \equiv\top?
\]
\end{minipage}%
}
\end{center}

\noindent SSC is known to be $\Sigma_2^p$-complete under polynomial-time
many-one reductions~\cite{Umans1999}.

\subsection{The reduction}

Let $(\theta_1,\ldots,\theta_n,k)$ be an instance of SSC.
Introduce fresh propositional variables $p_1,\ldots,p_n,q,r$ and put
\[
\Psi:=\Phi\cup\{p_1,\ldots,p_n,q,r\}.
\]
Intuitively, the variables $p_1,\ldots,p_n$ encode a subset of
$\{1,\ldots,n\}$. The role of the variables \(q,r\) will become clear later.

We define three auxiliary formulas in $\mathrm{PL}[\Psi]$.
\[
\eta := \bigvee_{i=1}^n \neg p_i,
\qquad
\tau := \bigvee_{i=1}^n (p_i\land\theta_i),
\qquad
\delta := \eta\land \tau.
\]
Finally, we define
\[\phi := \delta \lor (r \land (\neg q \lor \eta)).\]
The high-level idea now is the following. Suppose that \(\psi\) is equivalent to \(\phi\). Since every formula equivalent to $\phi$ must contain both $q$ and $r$, we can assign them sufficiently large weights so that \(\psi\) contains exactly one occurrence of each if it is within the prescribed weight bound. This in turn gives us control over the structure of \(\psi\). In particular, we can show that \(\psi\) can be rewritten, without increasing its weighted size, in the form
\[
\alpha\lor(r\land \beta),
\text{ where }
\alpha\equiv \delta
\text{ and }
\delta\lor \beta\equiv \neg q \lor \eta.
\]
Since $\alpha$ computes $\delta$, it must contain every variable $p_i$
both positively and negatively. We choose the weights of the variables \(p_i\) so
that $\beta$ can mention at most $k$
distinct variables from the set \(\{p_1,\dots,p_n\}\). We then show that the corresponding \(3\)-DNF formulas \(\theta_i\) form a cover.

Next we will define the weights. Let
\[b := \sum_{i=1}^n \mathrm{size}(\theta_i)\]
and set
\[
W:=b+1,\qquad M:=(3n+1)W.
\]
Define $w:\Psi\to\mathbb{Z}_+$ by
\[
w(s):=
\begin{cases}
1, & s\in\Phi,\\
W, & s\in\{p_1,\ldots,p_n\},\\
M, & s\in\{q,r\}.
\end{cases}
\]
Finally, put
\begin{equation}\label{eq:threshold}
K:=2M+(2n+k)W+b.
\end{equation}
The output of the reduction is the MEE instance $(\phi^w,K)$, which can be computed in polynomial time since the weights are polynomial in the size of the SSC instance. By Lemma~\ref{lem:leaf-weighting}, $L(\phi^w)=L_w(\phi)$.
Thus, to prove correctness, it suffices to show that the SSC instance is positive exactly when \(L_w(\phi) \leq K\).

\subsection{Left to right direction}

Suppose that there is a subset \(I \subseteq \{1,\dots,n\}\) with \(|I| \leq k\) such that
\[\bigvee_{i \in I} \theta_i \equiv \top.\]
Consider the formula
\[
\psi := \delta\lor\left(r\land
\left(\neg q \lor \bigvee_{i\in I}\neg p_i\right)\right).
\]
Notice that the only difference to \(\phi\) is that in the second disjunct we have \(\bigvee_{i \in I} \neg p_i\) instead of \(\eta\). Now
\begin{align*}
\mathrm{size}_w(\psi)
&= \mathrm{size}_w(\delta) + M + M + W|I| \\
&= nW + b + nW + M + M + |I|W \\
&\leq 2M +  (2n + k)W + b = K
\end{align*}
and so it suffices to show that \(\psi\) is equivalent to \(\phi\). Since
\[\bigvee_{i\in I}\neg p_i\models \eta,\]
we get \(\psi \models \phi\). For the converse direction, let $v$ be an assignment with $v(\phi)=1$. If $v(\delta)=1$, then $v(\psi)=1$. Otherwise, $v(\delta)=0$, so $v(r)=1$.
If $v(q)=0$, then $v(\psi)=1$. Therefore we may assume that $v(q)=1$, which in turn forces $v(\eta)=1$.
Since $v(\delta)=v(\eta \land \tau)=0$, we obtain $v(\tau)=0$, and hence
\begin{equation}\label{eq:1}
v(p_i\land\theta_i)=0
\qquad\text{for every }i\in\{1,\ldots,n\}.
\end{equation}
As $\bigvee_{i\in I}\theta_i\equiv\top$, there is some $i\in I$
with $v(\theta_i)=1$. By~\eqref{eq:1}, we then have $v(p_i)=0$.
Together with $v(r)=1$, this gives $v(\psi)=1$.
Thus $\phi\models\psi$, and consequently $\phi\equiv\psi$.

\subsection{Right to left direction}

Suppose that \(L_w(\phi) \leq K\). Thus there exists \(\psi\) such that \(\mathrm{size}_w(\psi) \leq K\) and \(\phi \equiv \psi\).

\medskip

\noindent\textbf{Claim.}\quad
\emph{The formula $\psi$ can be rewritten, without increasing its
weighted size, as}
\[
\alpha\lor(r\land \beta),
\]
\emph{where neither \(\alpha\) nor \(\beta\) contains
$r$ and furthermore}
\[
\alpha\equiv \delta
\qquad\text{and}\qquad
\delta\lor \beta\equiv \eta\lor\neg q.
\]

\begin{proof}
Recall that we may assume that $\psi$ is in negation normal form. Setting
$p_1=\cdots=p_n=1$ reduces $\phi$ to $r\land\neg q$, so $q$ must occur negatively in \(\psi\) and $r$ must occur positively in $\psi$. Since
$w(q)=w(r)>K/3$, both \(\neg q\) and \(r\) occur exactly once in \(\psi\).

Consider the syntax tree of \(\psi\). There exists a unique path from the unique occurrence of $r$ to the root
of $\psi$. Let $\gamma_1,\ldots,\gamma_\ell$ be the sibling subformulas
at the $\lor$-gates on this path, and let
$\sigma_1,\ldots,\sigma_m$ be the sibling subformulas at the
$\land$-gates. None of these subformulas contains $r$.

We first show that \(\neg\delta\models\neg\gamma_i\) for every \(i\).
Consider any assignment \(v\) such that \(v(\delta)=0\). Our goal is to
show that \(v(\gamma_i)=0\). Since \(q\) occurs only negatively in
\(\gamma_i\) and does not occur in \(\delta\), we may assume w.l.o.g.
that \(v(q)=0\). For any such assignment, we have
\(v[r/1](\psi)=1\) and \(v[r/0](\psi)=0\).
Suppose now that \(v(\gamma_i)=1\). Then the parent \(\lor\)-gate
of \(\gamma_i\) is true regardless of the value of \(r\).
Since \(r\) occurs nowhere else in \(\psi\), changing its value
cannot affect the truth value of \(\psi\). Consequently
\(v[r/0](\psi)=v[r/1](\psi)=1\), which is a contradiction.
Hence \(\neg\delta\models\neg\gamma_i\), which implies
\(\gamma_i\models\delta\).

Now define
\[
\alpha:=\bigvee_{i=1}^{\ell}\gamma_i,
\qquad
\beta:=\bigwedge_{j=1}^{m}\sigma_j,
\]
where an empty disjunction is $\bot$ and an empty conjunction
is $\top$. We claim that
\[
\psi\equiv \alpha\lor(r\land \beta).
\]
Indeed, if $\alpha$ is true, then some $\gamma_i$ is true, so $\delta$ and
therefore $\psi$ are true. If $\alpha$ is false, every sibling
at a $\lor$-gate on the path is false. These gates therefore
pass through the value of their child on the path, while
the $\land$-gates conjoin the sibling subformulas $\sigma_j$.
Thus, in this case, $\psi$ has value $r\land \beta$.

The sibling subformulas are disjoint and contain every leaf
of $\psi$ except its unique occurrence of $r$. Consequently,
\[
\mathrm{size}_w\bigl(\alpha\lor(r\land \beta)\bigr)
=
\mathrm{size}_w(\psi).
\]
Finally, setting $r=0$ gives $\alpha\equiv \delta$, while setting $r=1$
gives
\[
\delta\lor \beta
\equiv \alpha\lor \beta
\equiv \delta\lor \eta\lor\neg q
\equiv \eta\lor\neg q,
\]
where the last equivalence uses $\delta\models \eta$.
\end{proof}

Thus w.l.o.g. we can assume that \(\psi\) is of the form described in the above claim. Let
\[J := \{i\in\{1,\dots,n\} \mid p_i \text{ occurs in } \beta\}.\]
Since \(\alpha\equiv\delta\), the satisfiability of the formulas
\(\theta_i\) and the fact that \(n\geq 2\) imply that \(\alpha\) contains every variable \(p_i\) both positively and negatively. Hence \(\alpha\) contains at least \(2n\) occurrences of the variables \(p_1,\dots,p_n\), and \(\beta\) contains at least \(|J|\) of them. As \(\psi\) also contains one occurrence of \(q\) and one occurrence of \(r\), we obtain
\[2M+(2n+|J|)W \leq \mathrm{size}_w(\psi) \leq K = 2M+(2n+k)W+b.\]
Therefore \(|J|W\leq kW+b<(k+1)W\), and so \(|J|\leq k\).

It remains to show that
\[\bigvee_{j\in J}\theta_j\equiv\top.\]
Aiming for a contradiction, suppose that this is not the case. Then there is an assignment \(v\) such that \(v(\theta_j)=0\) for every \(j\in J\). Since the variables \(p_1,\dots,p_n,q\) do not occur in the formulas \(\theta_j\), we may assume that \(v(p_j)=1\) for every \(j\in J\), \(v(p_i)=0\) for every \(i\notin J\), and \(v(q)=1\). As \(|J|\leq k<n\), we have \(v(\eta)=1\). Moreover, \(v(\tau)=0\) by the choice of \(v\), and hence \(v(\delta)=0\). Since \(v(\neg q\lor\eta)=1\) and \(\delta\lor\beta\equiv\eta\lor\neg q\), we get \(v(\beta)=1\). Now let \(v'\) be obtained from \(v\) by setting \(p_i\) to \(1\) for every \(i\notin J\). These variables do not occur in \(\beta\), so \(v'(\beta)=v(\beta)=1\). On the other hand, \(v'(\neg q\lor\eta)=0\), which contradicts \(\delta\lor\beta\equiv\eta\lor\neg q\). Thus \(J\) is a cover of size at most \(k\).

\end{document}